\documentclass[11pt]{article}

\usepackage{fullpage}
\usepackage{booktabs} 

\usepackage{amsmath}
\usepackage{amssymb}
\usepackage{amsthm}
\usepackage{multirow}
\usepackage{hhline}
\usepackage{graphicx}
\usepackage{caption}
\usepackage{subfigure}

\newtheorem{lemma}{Lemma}
\newtheorem{theorem}{Theorem}

\newcommand{\ve}{{\varepsilon}}
\def\al{\alpha}
\def\bt{\beta}

\def\eps{\varepsilon}

\begin{document}
\title{Quantifying the Benefits of Infrastructure Sharing}
\author{Matthew Andrews\thanks{Bell Labs, Nokia.} \and Milan
  Bradonji\'c\thanks{Rutgers University. Work performed while at Bell
    Labs, Nokia.} \and Iraj Saniee\thanks{Bell Labs, Nokia.}}
\date{}
\maketitle
\begin{abstract}
We analyze the benefits of network sharing between telecommunications
operators. Sharing is seen as one way to speed the roll out of expensive
technologies such as 5G since it allows the service providers
to divide the cost of providing ubiquitous coverage. 
Our theoretical analysis focuses on scenarios with two service providers and
compares the system dynamics when they are competing with the dynamics 
when they are cooperating. We show that sharing can be
beneficial to a service provider even when it has the power to
drive the other service provider out of the market, a byproduct of
a non-convex cost function. 
A key element of this study is an analysis of the
competitive equilibria for both cooperative and non-cooperative 2-person games
in the presence of (non-convex) cost functions that involve a fixed
cost component. 
\end{abstract}

\section{Introduction}
\label{s:intro}
As communication technologies become increasingly complex, Service
Providers (SPs) need
to make ever larger investments in order to bring the latest network
generation to their end users.  As one way to defray these costs,
SPs are looking at network sharing agreements that allow them to
upgrade their networks more quickly at lower cost. For example, it has been observed that the most lucrative 10\% of
mobile access markets already account for over 50\% of an SP's
revenues whereas the remaining 90\% of the markets are ``subsidized'' by
those~\cite{Larsen12}.  Network sharing is especially attractive in the less
profitable regions since it minimizes the investment that SPs
need to make there.

Network sharing can take many forms depending on the assets that are
shared. In the most extreme case the entire network is shared. For
this case the only way in which the SPs can distinguish
themselves is via different service plans. The actual network
performance will be the same for both SPs. In less extreme cases
only parts of the network will be shared. Examples include one or more
of real-estate sharing, tower sharing, RAN (radio access network)
sharing and core network sharing. Sharing of such inactive elements is
becoming increasingly popular. In China the three major operators have
formed a joint venture to share the towers that host their radio
equipment~\cite{DengWW15}. 

In this work we focus on two service providers and derive a model to help understand the implications of
network sharing. 
We wish to understand both the implications for the profit of the
SPs as well as the prices faced by the end users. We note that network
sharing is a topic of interest to regulators since it can affect the
competitive make-up of a market. A regulator may wish to make sure
that prices do not rise excessively before signing off on a sharing
agreement. We incorporate the effects of such a regulator into our
analysis. 

The main question we ask is: {\em How does network pricing and
  capacity provisioning differ in the case that the SPs
  cooperate versus the case that they compete?}  In general, we show
that sharing can be beneficial for a wide range of network
parameters. In particular:
\begin{itemize}
\item We show that a sharing strategy can generate significantly higher profits for
  the service providers than if they compete and act according to a
  Nash equilibrium strategy.
\item We demonstrate that in some situations this gain due to sharing
  holds, even if a service provider has the market power to drive the
  other service provider out of the market. 
\end{itemize}
For the case of not sharing, we analyze the competition between the
service providers both for the case in which providers decide on how
much capacity to deploy (which gives rise to a Nash-Cournot game). In
the Appendix we also study 
the case in which providers decide on the price they offer to
the market (which gives rise to a Bertrand game). A key difficulty is
that in many situations service providers have a {\em fixed cost
  component} for entering a market which gives rise to a non-convex cost
function. Competition in this setting is non-standard and
leads to a potential situation in which one provider can drive the
other out of the market. Analyzing how this occurs is a key
component of our analysis.

\subsection{Problem variants}

\begin{figure}[htb]
\begin{center}
\includegraphics[scale = 0.30]{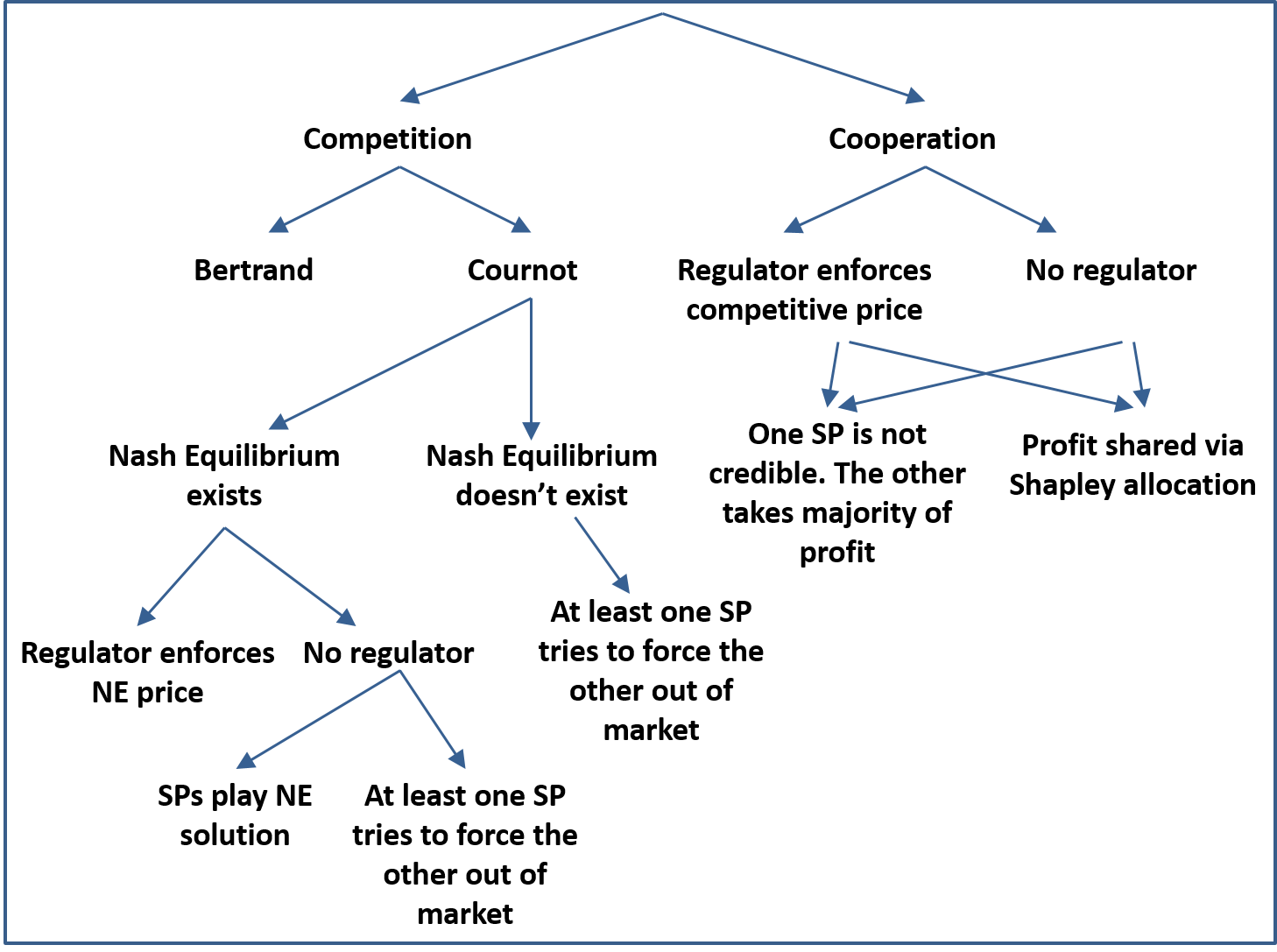}
\caption{A schematic of the problem variants that we consider.}
\label{f:hierarchy}
\end{center}
\end{figure}
Much of the paper will focus on the dynamics and outcomes of a given competitive or
collaborative scenario. Once those scenarios are established we can
determine whether or not an SP would be better off entering into a sharing
agreement or going it alone. However, there are many types of
competitive or cooperative arrangements that each lead to different
results for the SPs. We begin by giving a high-level description
of the different regimes that can arise. These can be
categorized in a hierachical manner as shown in
Figure~\ref{f:hierarchy}. 

At a high level each SP needs to determine whether to {\em compete}
with the other SP or whether to {\em cooperate} with it. For the case
that SPs decide to compete there are two types of competition,
a {\em Bertrand} game in which the SPs each offer their own price and
the end users react accordingly, and a {\em Cournot} game in which
each SP offers a certain capacity into the market and the market price
is set accordingly. 

The Cournot game admits a number of possible variations. First of all,
depending on the problem parameters, there may or may not be a {\em
  Nash equilibrium} solution for the SPs. If there is a Nash
equilibrium solution, then two situations can arise. In one of them
there is a {\em regulator} that forces the end user price to be the
one determined by the equilibrium solution. In the other there is no
regulator and so each SP has to decided whether to play the
equilibrium solution. The reason it might not do so is that depending
on the cost parameters an SP might have the ability to force the other SP
out of the market, i.e.\ it can play a strategy such that the other
SP has no incentive to participate in the market. However, that is a
risky strategy in that if both SPs try it then they could both be
worse off than playing the equilibrium solution.
We remark that the ability of one SP to drive the other out of the
market arises solely because of the fixed component in the cost that
leads to a non-convex cost function. 


Our remaining analysis considers the situation when SPs
cooperate and share their networks. For this case there are variants
depending on whether there is a regulator that can enforce price limits. If
there is a regulator then we assume that the end user price is forced
to be the same as it would be if the SPs were competing according to a
Cournot game. If there is no such regulator then we assume that the
SPs can set prices as if the combined entity were a monopoly. 

The remaining decisions refer to how the profits are shared when the
two SPs cooperate. We assume that this depends on whether or not both
SPs are {\em credible players}, i.e.\ whether or not they can
legitimately offer service by themselves. If one SP deems that the
other is not credible, then it would look to take the vast majority of
the profit for itself. If both SPs
can credibly offer service on their own, then we assume that profits
in the sharing scenario are divided up according to appropriate schemes,
(e.g.\ based on Shapley value). 

\subsection{Paper Organization}
In Section~\ref{s:model} we define the models that we
  use throughout the paper and we outline the equations that define optimal
  demand and price for a monopolistic SP.
In Section~\ref{s:narrative} we explore the dynamics for 
a prototypical example in a single geographic region. By starting with a concrete
example we can better understand the fundamental dynamics at play
rather than getting bogged down in the general equations (which can
become somewhat complex). For this example we consider most of the
regimes outlined in Figure~\ref{f:hierarchy}. 

In Section~\ref{s:general} we outline our results for the general
case with much of the detailed
analysis deferred to the Appendix. In particular, in Appendix~\ref{s:cournot} we present a more detailed study of the Cournot game
  in which the providers compete via deployed capacity. In particular,
  we derive the equations that characterize the Nash Equilibrium if it
  exists under our non-convex cost functions. We also examine the conditions under which one SP can drive
  the other out of the market. 
In Appendix~\ref{s:sharing} we determine the SP profits that arise when
they enter into a sharing agreement. These profits are determined by
how the combined profit is shared. This is calculated via two notions
of Shapley value (one of which is based on the notion of ``Shapley
value with externalities''.)
We also explore
  a regulatory framework that enforces Nash prices so that end users
  are not
  penalized in the sharing scenario. Lastly, we examine how an SP
  would approach a sharing agreement if it does not deem the other SP
  to be a credible competitor, i.e.\ if it does not believe the other
  SP has the cost structure to operate a network on its own. 

In Appendix~\ref{s:bertrand-single} we analyze the dynamics when the
competitive setting is modeled as a Bertrand game. The analysis of the
Bertrand game is in general simpler than the corresponding Cournot
analysis. This is because in a Bertrand game the SP with the
better cost structure always has the ability to drive the other SP out
of the market. In addition to the
basic Bertrand game we study a number of variants.  In one of them the
SPs are able to share network costs but they must still compete
on price. In another variant only a subset of the end users are deemed
to be price conscious. 

In Appendix~\ref{s:narrative-multiple} we extend our analysis to
a situation where {\em a priori} the SPs offer service in different geographic regions. This is an especially attractive
situation for network sharing since it allows each SP to 
offer service over the entire market more rapidly.

\subsection{Previous Work}
The GSM Association wrote an influential report~\cite{GSMA} examining
the ways in which wireless infrastructure can be shared. This report
describes some existing sharing agreements that are already in place
and discusses the economic and regulatory implications.  In
\cite{JanssenLS14}, Janssen et al.\ discuss the statistical multiplexing gains that
can be obtained by combining capacity in a network sharing
arrangement. The economics of the Chinese tower sharing agreement
mentioned earlier were analyzed by Deng et al. in \cite{DengWW15}.  
Malanchini and Gruber studied small cell sharing in
\cite{MalanchiniG15} and presented ways in which operators could still
differentiate themselves (e.g.\ via power management) even if all
network resources are shared.  The papers
\cite{AlQahtani08,KokkuMZR12,MalanchiniVA14,ValentinJA13} discuss ways
in which network sharing could be realized in practice. In particular,
\cite{KokkuMZR12} discusses a technique known as ``network slicing''
in which the resource allocation algorithms at wireless basestations
reserve a fraction of resources for each SP. 
The general economics literature contains many analyses of duopolies
with various cost structures (e.g.\ \cite{Saporiti08}). However, to the best of our knowledge
previous work has not considered Cournot competition for network
providers under non-concave
cost functions with fixed costs, and there has not been a comparison of
how the dynamics under sharing compare to the competitive dynamics. 

\section{Cost Model}
\label{s:model}

We consider two Service Providers (SPs) that we denote SP1 and SP2. We
begin with a single geographical region in which both SPs operate. If
SP $i$ serves demand\footnote{Throughout this paper we employ a coarse
  measure of demand, namely bytes per month across the whole
  region. Although this is coarse, the expenses of an operator are
  closely tied to that number. We also assume that demand is based on
  a single price for each operator. In reality each operator
  offers multiple data plans with different sizes and costs. We leave
  the incorporation of different data plans into our analysis as an
  interesting direction for future work.} of size $q_i$ then its cost is given by,
$$
C_i(q_i)=\left\{\begin{array}{cc}\alpha_i+\beta_iq_i&q_i>0\\0&q_i=0\end{array}\right.,
$$
for some parameters $\alpha_i>0,\beta_i\ge 0$. The parameter $\alpha_i$
reflects a fixed cost, e.g.\ the cost of infrastructure such as
buildings or spectrum, whereas
the parameter $\beta_i$ reflects the cost of serving a unit of demand,
e.g.\ the cost of deploying equipment. 

The level of demand in the region (denoted $q$) is closely related to
the price offered to the end users (denoted $p$). Our
assumption is that all users in the region are price conscious.
In particular, we assume that the market is elastic with elasticity
coefficient $\ve$ greater than 1, i.e.\ 
$
q(p)=Qp^{-\ve}, 
$
for some parameter $Q$ (quantity sold at unit price).
Roughly speaking this means that a 1\%
reduction in price results in an $(\ve-1)$\% increase in revenue. 
We note that a change in the demand can reflect both a change in the
number of end users creating that demand as well as a change in the
demand per end user. 

If SP $i$ serves demand $q_i>0$ at price $p_i$ then it receives a profit
given by, $\Pi_i=q_ip_i-C_i(q_i)
=q_ip_i-(\alpha_i+\beta_iq_i)$.
We observe that the fixed cost $\alpha_i$ makes
the cost function non-convex which distinguishes our analysis from
many previous studies of network economics. In particular, 
if $\alpha_i\gg 0$ then SP $i$ may not wish to participate in the
market because even as a monopolist it is unable to make a profit. If this decision is due to the capacity (Cournot) or price
(Bertrand) offered by the other SP then we say that
SP $i$ is {\em driven out of the market}. 

\begin{lemma}
\label{l:monopoly}
Under monopoly pricing we have,
\begin{eqnarray*}
p^{mon} &=& \ve \beta/(\ve-1),~~~~~~~~q^{mon} = Q(\ve \beta/(\ve-1))^{-\ve}\\
\Pi^{mon} &=& Q\left((\frac{\ve\beta}{\ve-1})^{1-\ve}-\beta(\frac{\ve\beta}{\ve-1})^{-\ve}\right)-\alpha.
\end{eqnarray*}
assuming that $\Pi^{mon}\ge 0$. (If not then the SP stays out of the market.)
\end{lemma}
The proof (which is standard) is given in Appendix~\ref{s:monopoly}.

\section{Narrative for a single example}
\label{s:narrative}

We begin by exploring the dynamics for 
a prototypical example. In this way
we can better understand the fundamental dynamics at play
rather than getting bogged down in the general equations (which can
become somewhat complex). 
In later 
sections we consider the general case (with much of the proofs and
derivations 
deferred to the Appendix). 

For our example the unit of demand is a PB and the unit of price is
\$1M. For these units the price elasticity function has parameters $\ve=1.25$ and
$Q=1000$.  The per-unit capacity costs are $\beta_1=\$2.5M$ for SP1
and $\beta_2=\$2M$ for SP2 per petabyte (PB) of wireless capacity.  The fixed
capacity costs (representing the cost of participating in the market, e.g., 
for buying spectrum or building cell towers) are $\alpha_1
= \$50M$ and $\alpha_2 = \$100M$.
For these parameters, the monopoly price, demand and profit for each
operator is given by,
\begin{eqnarray*}
p_1^{mon}=\$12.5M,&~&q_1^{mon}=42.6PB,~\Pi_1^{mon}=\$376M\\
p_2^{mon}=\$10.0M,&~&q_2^{mon}=56.2PB,~\Pi_2^{mon}=\$350M
\end{eqnarray*}

\subsection{Network Sharing}
\label{s:narrative-sharing}

We first examine the situation under network sharing. 
This is the easiest case to consider since we do not need to worry
about the competitive dynamics between the SPs. In particular the
SPs cooperate and use the lowest cost parameters that are available,
i.e.\
$$
C^{coop}(q)=\alpha_{min}+\beta_{min}q,
$$
where $\alpha_{min}=\min\{\alpha_1,\alpha_2\}$ and
$\beta_{min}=\min\{\beta_1,\beta_2\}$. (See Figure~\ref{f:sharingcost}
for a depiction of $C^{coop}(q)$ in comparison to $C_1(q)$ and
$C_2(q)$.)
\begin{figure}[htb]
\begin{center}
\includegraphics[width = 3.in]{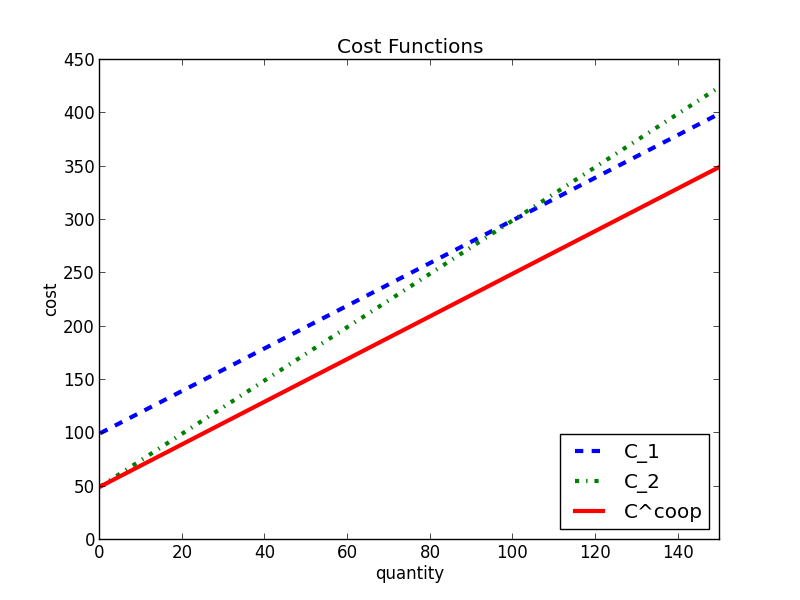}
\caption{The cost function $C^{coop}(q)$ compared against $C_1(q)$ and
  $C_2(q)$.}
\label{f:sharingcost}
\end{center}
\end{figure}

We assume that the combined entity is able to use monopoly
pricing. In this case the combined price, demand and profit for our
running example is given
by,
$$
p^{coop}=\$10.0M,~
q^{coop}=56.2PB,~
\Pi^{coop}=\$400M
$$

It remains to determine how the profit is split between the SPs. A
natural way to do this is via the Shapley value which gives to SP $i$
its expected contribution to the coalition assuming that the SPs
create the coalition in a random order. If we assume that the first SP
to enter the coalition can utilize monopoly pricing then the profits
are given by,
\begin{eqnarray*}
\Pi_1^{coop}&=&\frac{1}{2}(\Pi_1^{mon}+\Pi^{coop}-\Pi_2^{mon})=\$213M,\\
\Pi_2^{coop}&=&\frac{1}{2}(\Pi_2^{mon}+\Pi^{coop}-\Pi_1^{mon})=\$187M.
\end{eqnarray*}

In order for an SP to determine whether network sharing is the best option,
it needs to compare its profits under sharing with its profits
for the case in which it competes with the other SP.  As mentioned in
Section~\ref{s:intro}, there are many notions of competition - a
Cournot game in which the SPs offer capacity to the market and the
market sets the price, and a Bertrand game in which the SPs directly
offer prices to the market. We begin by examining the Cournot game and
defer the corresponding analysis of the Bertrand game to the
Appendix. 

\subsection{Cournot Game}
\label{s:cournot-narrative}

\begin{figure*}[htb]
\begin{center}
\includegraphics[width=2.0in]{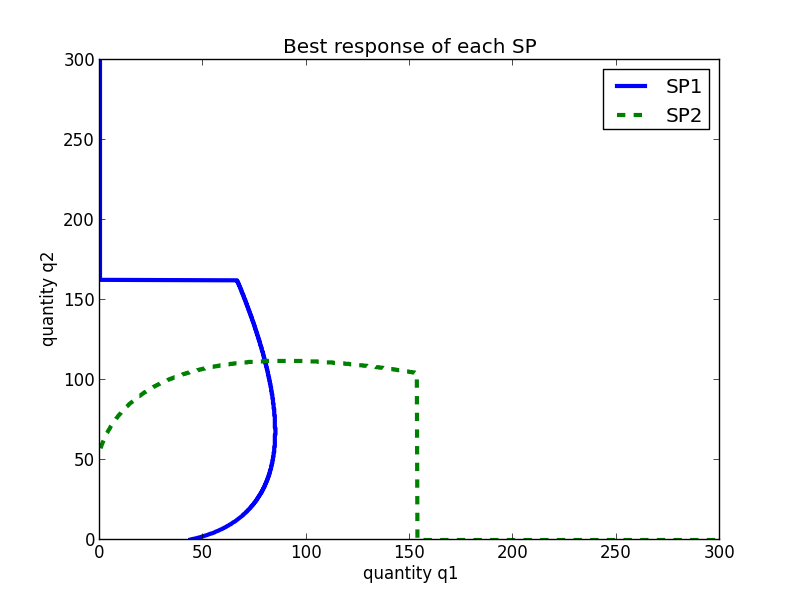}
~~
\includegraphics[width=2.0in]{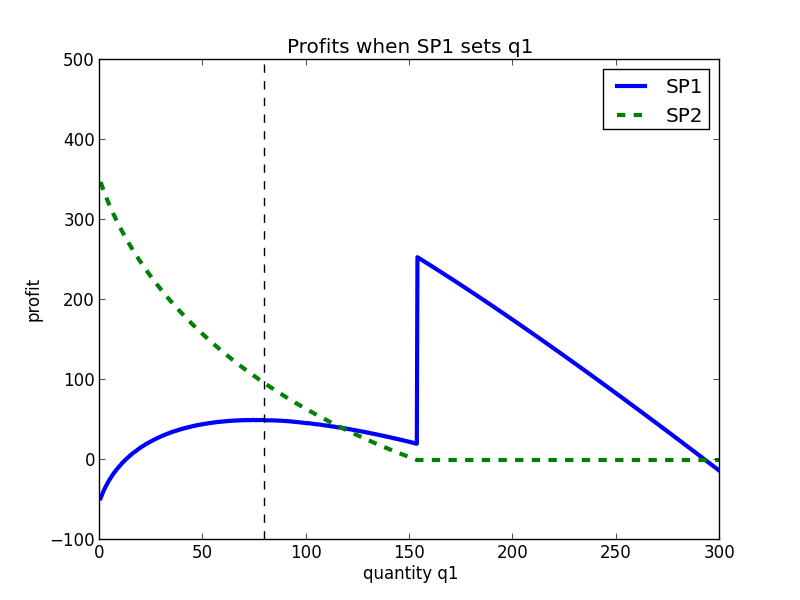}
~~
\includegraphics[width=2.0in]{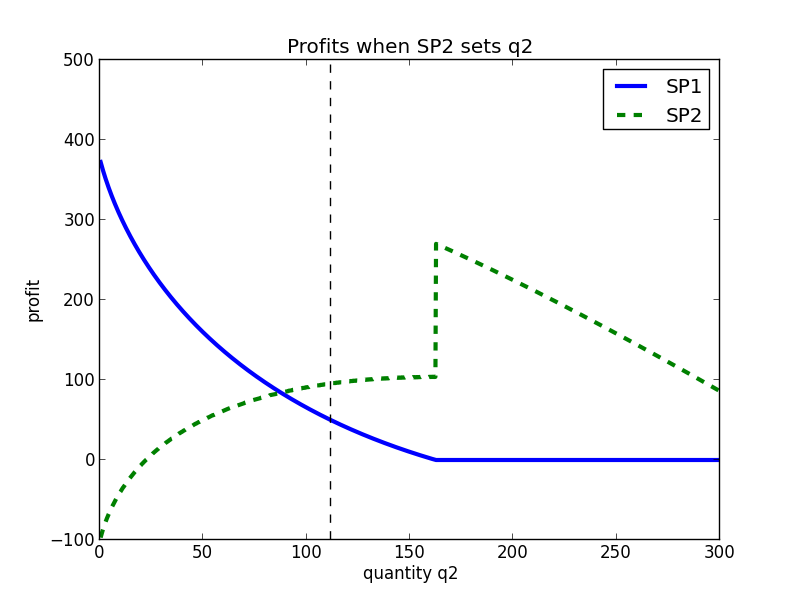}
\caption{(Left) The best response for each SP given the actions of the
  other SP. (Middle) The profits of the two SPs as a function of $q1$
  (assuming SP2 plays the best response).
(Right) The profits of the two SPs as a function of $q2$ (assuming SP1
plays the best response). The vertical lines represent the
Nash-Cournot solution ($q_1=80$, $q_2=112$). 
}
\label{f:profits-q1q2}
\end{center}
\end{figure*}
In the Cournot game SP $i$ offers to serve demand $q_i$ and then the
market determines a common price $p$ based on $q_1$ and $q_2$. In
particular,
$$
p=\left(\frac{q_1+q_2}{Q}\right)^{-1/\ve}.
$$
The main complication with the Cournot game in our setting is the
presence of the $\alpha_i$ terms that introduce a discontinuity into
the profit functions. If SP $i$ cannot generate a positive profit for 
any value $q_i>0$ then it will simply set $q_i=0$ and accept zero
profit. In this case we say that SP $i$ is {\em driven out of the
  market}. 

Figure~\ref{f:profits-q1q2} (left) illustrates the behavior of the
Cournot game. The plot shows the best response for each SP,
given the actions of the other SP. In particular, for any value $q_1>0$
(on the x-axis),
the green dashed curve represents the value of $q_2$ (on the y-axis) that
maximizes the profit of SP2 assuming that the action of SP1 is
$q_1$. 

The blue solid curve is similar and represents the best response for SP1 but with the axes flipped. 
In particular, for any value $q_2$
(on the y-axis),
the blue curve represents the value of $q_1$ (on the x-axis) that
maximizes the profit of SP1 assuming that the action of SP2 is
$q_2$. 

The crossing point of the curves at $(q_1,q_2)=(80,112)$ represents a
Nash equilibrium. (A closed form expression for this equilbrium is
presented in Appendix~\ref{s:cournot}.) In other words, if the action
of SP1 is $q_1=80$ then the optimal action of SP2 is $q_2=112$ and
if the action of SP2 is $q_2=112$ then the optimal action of SP1 is
$q_1=80$.  At this point the full set of quantity, price and profit
values is given by,
\begin{eqnarray*}
p_1^{NC}=\$3.75M,&~&q_1^{NC}=80PB,~\Pi_1^{NC}=\$50M\\
p_2^{NC}=\$3.75M,&~&q_2^{NC}=112PB,~\Pi_2^{NC}=\$96M.
\end{eqnarray*}
 
\subsection{The impact of a regulator}

As we will discuss shortly, the profits for each SP in the sharing
scenario are significantly higher than they are in the competitive
scenario. One main reason for this is that in the sharing
case the combined entity is able to set a monopoly price. The profit
increase therefore comes at the expense of the end users who have to
pay higher prices. A regulator may deem this to be anti-competitive
and as a result may impose an upper bound on the price that the
combined entity can charge if the two SPs decide to share. A natural candidate for this upper bound
is the price corresponding to the Nash Equilibrium that we computed in
the previous section. Sharing can still be
beneficial even if a regulator restricts the price because the combined
entity can take advantage of the reduced cost function. (See
Figure~\ref{f:sharingcost}.) Moreover, in contrast to the full
competitive case of Section~\ref{s:cournot-narrative} the SPs will share profits according to Shapley
value and so we have, 
\begin{eqnarray*}
p^{reg}&=&p_1^{NC}=\$3.75M,~q^{reg}=q_1^{NC}+q_2^{NC}=192PB,\\
\Pi^{reg}&=&p^{reg}q^{reg}-(\alpha_{min}+\beta_{min}q^{reg})=\$286M
\end{eqnarray*}
This combined profit is shared via Shapley value and results in,
$$
\Pi_1^{reg}=\$120M,~\Pi_2^{reg}=\$166M
$$

\subsection{Comparison of network sharing and the Cournot game}

In the table in Figure~\ref{f:simplepayofftable} we summarize the
price $p$ and the profits $\Pi_1,\Pi_2$ for the case that the SPs
compete according to the Nash Equilibrium of the Cournot game (the
so-called Nash-Cournot solution), as well as the cases of network
sharing with and without a regulator. We see that for the case of
network sharing without a regulator, both SPs have significantly
higher profit than in the Nash-Cournot solution but this is partly
because they can charge a monopoly price to the end users. In the case that a regulator enforces an upper bound on price
equal to the Nash-Cournot price, both SPs still obtain a higher profit
with network sharing than in the Nash-Cournot solution. This latter
effect comes from the fact that the SPs can share the cost of the
network and utilize the cost parameters $\alpha_{min},\beta_{min}$
rather than $\alpha_i,\beta_i$. 

\begin{figure}[htb]
\centering
\begin{tabular}{|c|c|c|}\hline 
& price per PB, $p$ & profits $(\Pi_1,\Pi_2)$ \\ \hline
Nash-Cournot&3.75&(50, 96)\\ \hline
Sharing (no regulator)&10&(213,187)\\ \hline
Sharing (regulator)&3.75&(120,166)\\ \hline
\end{tabular}
\caption{The price $p$ and the SP profits $(\Pi_1,\Pi_2)$ due to the
  Nash-Cournot solution and the sharing scenarios. All quantities are
  in units of \$1M.}
\label{f:simplepayofftable}
\end{figure}

\subsection{Aggressive and submissive strategies}
\label{s:aggressive-submissive}

We now address the more complex dynamics that can arise due to the
fact that $\alpha_i>0$. In particular, we ask whether the SPs will be
motivated to conform to the Nash-Cournot solution or whether they might be tempted to deviate from that
action. Note that both the blue and the green curves
hit zero in Figure~\ref{f:profits-q1q2} (left), i.e.\ both SPs have the
ability to drive the other one out of the market. 

Figure~\ref{f:profits-q1q2} further illustrates 
why the temptation to deviate from the Nash equilibrium might exist. 
In particular,
Figure~\ref{f:profits-q1q2} (middle) shows the profits of the two SPs given a
fixed value of $q_1$. More precisely, for each value of $q_1$, the
solid blue and dashed green curves show the profits of SP1 and SP2 respectively,
if SP2 plays its optimal response. We see that there is a big
discontinuity in the profit of SP1. At the point ($q_1=154$) at
which SP1 drives SP2 out of the market, the optimal response of SP2
jumps from a non-zero value to zero. This means there is less
capacity available which in turn drives up the price and hence the
profit of SP1. Note however that it cannot claim the monopoly profit
since it cannot reduce to the monopoly point $q_1=42.6$ without
letting SP2 back into the market. 

However, Figure~\ref{f:profits-q1q2} (right) shows that this process could work
the other way round as well. In particular,
this figure shows the profits of the two SPs given a
fixed value of $q_2$. For each value of $q_2$, the
solid blue and dashed green curves show the profits of SP1 and SP2 respectively,
if SP1 plays its optimal response. This time we see a big jump in the
profit of SP2 at the point at which it drives SP1 out of the
market. 

As a result, each SP could benefit if it sets its $q_i$ at a level
that drives the other SP out of the market {\em and the other SP
  acquiesces to being driven out}.  We can therefore model the game by
assuming that the SPs have the following discrete choices. 
\begin{itemize}
\item Nash-Cournot: In this case each SP assumes that the other SP
  will compete in the market in which case it makes sense to play the
  Nash equilibrium value. 
\item Aggression: An aggressive SP will play at a level that drives the
  other SP out of the market.
\item Submission: A submissive SP will accept being shut out of the
  market (rather than fighting it and potentially taking a loss).
\item Sharing: An SP can offer to enter a sharing arrangement with the
  other SP. However, the arrangement only goes into effect if
  both SPs agree to it. If both SPs do agree then we obtain the
  profits presented in Section~\ref{s:narrative-sharing} (that depend
  on whether or not a regulator caps the price $p$).
\end{itemize}


\begin{figure*}[htb]
\centering
\begin{tabular}{|c|c|c|c|c|c|}\hline 
& & & & Sharing & Sharing \\
(SP1 profit,SP2 profit) & Nash-Cournot & Aggression & Submission & (no
regulator) & (regulator) \\ \hline 
Nash-Cournot &(50, 96)&(-1,80)&({\em 353,0})& NA & NA\\ \hline
Aggression&(9,-1)&(-48,-17)&(254,0)&NA & NA\\ \hline
Submission&({\em 0,321})&(0,271)&(0,0)&NA & NA\\ \hline
Sharing (no regulator) &NA&NA&NA&(213,187) &NA\\ \hline
Sharing (regulator) &NA&NA&NA&NA&(120,166)\\ \hline
\end{tabular}
\caption{The profits $(\Pi_1,\Pi_2)$ due to the different strategy
  combinations. The rows represent the decisions for SP1 and the
  columns represent the decisions for SP2.}
\label{f:payofftable}
\end{figure*}
The table in Figure~\ref{f:payofftable} shows the outcomes of the resulting game. 
The columns represent the decisions for
SP1 and the rows represent the decisions for SP2. The entries have
the form (SP1 profit,
SP2 profit). 
The profits due to sharing are included in the bottom right corner of
the table. Since sharing only takes place if both parties agree to
share, there is no entry in the table if only one SP is sharing.

We remark that the italicized entries (where one SP plays the
Nash-Cournot strategy and the other plays the Submission strategy)
are not viable outcomes because the submissive SP would be better off
playing the Nash-Cournot strategy as well. 
From the above we can conclude the following. 
\begin{itemize}
\item Network sharing is better for both SPs than the Nash-Cournot solution
\item If sharing does not take place then the Nash-Cournot solution 
  is a unique Nash equilibrium since that is the only point at which
  both curves cross in Figure~\ref{f:profits-q1q2} (left). 
\item An SP might be tempted to deviate from the Nash-Cournot solution
  since if it aggressive then the best response of the
  other SP is to be submissive in which case the aggressive SP will do
  even better than sharing. (We note that this is not a Nash
  equilibrium since if one SP sets $q_i=0$ then the best response of
  the other SP is to set $q_{3-i}=q_{3-i}^{mon}$.)
\item The downside of aggression is that if both SPs are
  aggressive then they both make negative profit and hence are worse
  off (as in standard games of chicken). Whether or not an SP will
  choose to be aggressive will largely depend on how it expects the
  other SP to react. 
\end{itemize}

\section{General Analysis}
\label{s:general}

As mentioned earlier, most of our general analysis is deferred to the
Appendices. However, in the following we state
our main results for the case of network sharing compared with a
Cournot competition. 

For the case of sharing the combined price, demand and profit is given
by, 
\begin{eqnarray*}
p^{coop}&=&\ve \beta_{min}/(\ve-1)\\
q^{coop}&=&Q(\ve \beta_{min}/(\ve-1))^{-\ve}\\
\Pi^{coop}&=&p^{coop} q^{coop} - (\alpha_{min}+\beta_{min} q^{coop}),
\end{eqnarray*}
where $\alpha_{min}=\min\{\alpha_1,\alpha_2\}$ and
$\beta_{min}=\min\{\beta_1,\beta_2\}$. 
The profit is shared according to,
\begin{eqnarray}
\Pi_1^{coop}&=&\frac{1}{2}(\Pi_1^{mon}+\Pi^{coop}-\Pi_2^{mon})\label{eq:shareprofit1}\\
\Pi_2^{coop}&=&\frac{1}{2}(\Pi_2^{mon}+\Pi^{coop}-\Pi_1^{mon})\label{eq:shareprofit2},
\end{eqnarray}
where $\Pi^{mon}_1,\Pi^{mon}_2$ are the monopoly profits for SPs 1 and
2 respectively. 

For the case of the Cournot competition, for any given $q_1,q_2$ the
profits of the SPs are specified:
\begin{eqnarray}
\Pi_1(q_1,q_2)&=& \left(\frac{Q}{q_1+q_2}\right)^{1/\eps}
q_1 - \alpha_1 -\beta_1 q_1 \label{eq:prof1}\\
\Pi_2(q_1,q_2)&=& \left(\frac{Q}{q_1+q_2}\right)^{1/\eps}
q_2 - \alpha_2 -\beta_2 q_2 \label{eq:prof2}
\end{eqnarray}
We first need to calculate the best response function for each
provider. If SP2 offers to serve demand $q_2$, we show that the best
response of SP1 is to set $q_1$ so that $q_1/q_2$ is the solution to the equation,
$$
(z+1)^{1+\frac{1}{\eps}}=Az+B,
$$
where $A=(1-\frac{1}{\eps})(Q/(q_2\beta_1^{\eps}))^{1/\eps}$ and
$B=(Q/(q_2\beta_1^{\eps}))^{1/\eps}$, assuming that this leads to a
positive profit. We use $\hat{q}_1(q_2)$ to denote this solution for any fixed value of
$q_2$. The best response function for SP2 can be defined analogously. 

Now that the best response functions are in place, we can calculate
the Nash-Cournot solution which is given by,
\begin{eqnarray*}
t&=&\frac{1-\beta_2/\beta_1
  (1-1/\eps)}{\beta_2/\beta_1-(1-1/\eps)}\\ 
q_1^*&=&
Q\left(\frac{1+t(1-\frac{1}{\eps})}{\beta_2(1+t)^{1+\frac{1}{\eps}}}\right)^{\eps}\\
q_2^*&=&Q\left(
\frac{\frac{1}{t}(1-\frac{1}{\eps})+1}{\beta_1(\frac{1}{t}+1)^{1+\frac{1}{\eps}}}
\right)^{\eps}\\
\end{eqnarray*}
(We remark that in some cases this Nash Equilibrium may not exist if either
$q^*_1$ or $q^*_2$ is negative or if either of the corresponding
profits are negative.) 

As we saw in our running example, even if the Nash Equilibrium does
exist an SP may have an incentive to not
play the Nash equilibrium solution but instead try to drive the other
SP out of the market. If SP1 wishes to be aggressive in this way then
it sets $q_1=q'_1$, where
$$
q'_1=\arg\max_{q_1:\Pi_2(q_1,\hat{q}_2(q_1))\le 0}\{\Pi_1(q_1,0)\}.
$$
(A similar expression can be derived for the aggressive strategy of
SP2.) Since for the submission strategy SP $i$ simply
sets $q_i=0$, we have now defined the values of $q_1,q_2$ for the
cases of Nash-Cournot, Aggression and Submission.  For any problem
instance we can then utilize
the profit functions of Equations (\ref{eq:prof1}) and
(\ref{eq:prof2}) for each possible pair of strategies and combine the
results with the profits for sharing (given in Equations
(\ref{eq:shareprofit1}) and (\ref{eq:shareprofit2}) in order to obtain a table of the
form shown in Figure~\ref{f:payofftable}. 

\section{Conclusions}
In this paper we have presented a model to illustrate the options
facing two Service Providers who are deciding whether or not to share
network infrastructure. Our cost function has a fixed
component and hence is non-convex. We presented a taxonomy of problem
variants and derived the profit for each SP both in the case that
they share infrastructure as well as the case that they are
competitors. For the latter case the dynamics are complicated because
the fixed cost function gives an SP the option to be aggressive and
try to drive the other SP out of the market. For our running example
the profit to each SP in the case of sharing is significantly higher
than when they compete fairly according to a Nash Equilibrium, and the
profit is comparable to the case in which the SP is aggressive and the
other SP is submissive. For this example each SP would likely consider
sharing to be a better option since even if it plays agressively in a
competitive setting there is no way to ensure that the other SP would
not try to be agressive as well. 

A number of problems remain. In particular we would like to determine
how the dynamics would change if the two SPs only share a portion of
the network infrastructure. We would also like to incorporate roaming
agreements into the model since they represent a ``halfway'' point
between total competition and full cooperation. 

\bibliographystyle{abbrv}
\bibliography{sharing}

\begin{thebibliography}{10}

\bibitem{AlQahtani08}
S.~AlQahtani.
\newblock Adaptive rate scheduling for {3G} networks with shared resources
  using the generalized processor sharing performance model.
\newblock {\em Computer Communications}, 31(1):103--111, January 2008.

\bibitem{BagwellL14}
K.~Bagwell and G.~Lee.
\newblock Number of firms and price competition.
\newblock http://web.stanford.edu/\~{}kbagwell/papers/
  Bagwell\%20Lee\%20s\%20021214.pdf, 2014.

\bibitem{DengWW15}
X.~Deng, J.~Wang, and J.~Wang.
\newblock How to design a common telecom infrastructure by competitors
  individually rational and collectively optimal.
\newblock In {\em Proc. of 4th IEEE Workshop on Smart Data Pricing (at
  Infocom)}, 2015.

\bibitem{GSMA}
{GSM Association:}.
\newblock Mobile infrastructure sharing.
\newblock http://www.gsma.com/publicpolicy/wp-content/
  uploads/2012/09/Mobile-Infrastructure-sharing.pdf.

\bibitem{JanssenLS14}
T.~Janssen, R.~Litjens, and K.~W. Sowerby.
\newblock On the expiration date of spectrum sharing in mobile cellular
  networks.
\newblock In {\em 12th International Symposium on Modeling and Optimization in
  Mobile, Ad Hoc, and Wireless Networks, WiOpt 2014, Hammamet, Tunisia, May
  12-16, 2014}, pages 490--496, 2014.

\bibitem{KokkuMZR12}
R.~Kokku, R.~Mahindra, H.~Zhang, and S.~Rangarajan.
\newblock {NVS: A} substrate for virtualizing wireless resources in cellular
  networks.
\newblock {\em IEEE/ACM Transactions on Networking}, 20(5):1333--1346, October
  2012.

\bibitem{MalanchiniG15}
I.~Malanchini and M.~Gruber.
\newblock How operators can differentiate through policies when sharing small
  cells.
\newblock In {\em {IEEE} 81st Vehicular Technology Conference, {VTC} Spring
  2015, Glasgow, United Kingdom, 11-14 May, 2015}, pages 1--5, 2015.

\bibitem{MalanchiniVA14}
I.~Malanchini, S.~Valentin, and O.~Aydin.
\newblock An analysis of generalized resource sharing for multiple operators in
  cellular networks.
\newblock In {\em Proceedings of IEEE Symposium on Personal, Indoor and Mobile
  Radio Commununication (PIMRC)}, September 2014.

\bibitem{MichalakRMSJ10}
T.~Michalak, T.~Rahwan, D.~Marciniak, M.~Szamotulski, and N.~Jennings.
\newblock Computational aspects of extending the shapley value to coalitional
  games with externalities.
\newblock In {\em ECAI}, pages 197--202, 2010.

\bibitem{Myerson77}
R.~Myerson.
\newblock Values of games in partition function form.
\newblock {\em International Journal of Game Theory}, 6(1):23--31, 1977.

\bibitem{Saporiti08}
A.~Saporiti and G.~Colomaz.
\newblock Bertrand's price competition in markets with fixed costs, 2008.
\newblock University of Rochester, Working Paper No. 541.

\bibitem{Larsen12}
D.~Telekom and K.~K. Larsen.
\newblock The ultra-efficient network factory: network sharing and other means
  to leapfrog operator efficiencies.
\newblock {\em Broadband MEA}, 2012.

\bibitem{ValentinJA13}
S.~Valentin, W.~Jamil, and O.~Aydin.
\newblock Extending generalized processor sharing for multi-operator scheduling
  in cellular networks.
\newblock In {\em Proceedings of the International Wireless Communication And
  Mobile Computing Conference (IWCMC)}, July 2013.

\bibitem{Varian80}
H.~Varian.
\newblock A model of sales.
\newblock {\em American Economic Review}, 70(4):651--659, 1980.

\end{thebibliography}

\appendix
\section{Detailed analysis of the Cournot competition}
\label{s:cournot}

We now present our more detailed general analysis. We start
with the competitive setting governed by the Cournot game, in the
which SPs compete by deciding how much
demand they wish to serve. Our analysis is more complex than the
textbook Cournot analysis since the presence of the non-zero
$\alpha_i$ parameters means that each SP is faced with a decision
regarding whether or not to compete. 

In the Cournot setting the price is determined by the
aggregate demand, i.e.\
$$
p^{NC}=\left(\frac{q_1^{NC}+q_2^{NC}}{Q}\right)^{-1/\ve},
$$
We first consider a relaxation of the game in which the $q_i$
values can be negative and the profit is always given by
$\Pi_i=q_ip_i-(\alpha_i+\beta_iq_i)$, regardless of whether or not it
is positive.  In this case we assume that the quantities $q_1^{NC}$, $q_2^{NC}$ are chosen so that they form a
Nash equilibrium with respect to the SP profits.  Hence we wish to
find a solution for which $\partial \Pi_i/\partial q_i=0$ for
$i=1,2$. 
\begin{theorem}
When the SPs compete on quantity in the relaxed Nash-Cournot game, the solution
is given by, 
\begin{eqnarray*}
t&=&\frac{1-\beta_2/\beta_1
  (1-1/\eps)}{\beta_2/\beta_1-(1-1/\eps)}\\ 
q_1^*&=&
Q\left(\frac{1+t(1-\frac{1}{\eps})}{\beta_2(1+t)^{1+\frac{1}{\eps}}}\right)^{\eps}\\
q_2^*&=&Q\left(
\frac{\frac{1}{t}(1-\frac{1}{\eps})+1}{\beta_1(\frac{1}{t}+1)^{1+\frac{1}{\eps}}}
\right)^{\eps}\\
p^{NC}&=& \left(\frac{q_1^*+q_2^*}{Q}\right)^{-1/\eps}\\
\Pi_1^{NC}&=&\Pi_1(q_1^*,q_2^*)= \left(\frac{Q}{q_1^*+q_2^*}\right)^{1/\eps}
q_1^* - \alpha_1 -\beta_1 q_1^*\\
\Pi_2^{NC}&=&\Pi_2(q_1^*,q_2^*)= \left(\frac{Q}{q_1^*+q_2^*}\right)^{1/\eps}
q_2^* - \alpha_2 -\beta_2 q_2^*
\end{eqnarray*}
\end{theorem}
\begin{proof}
For ease of notation we drop the superscript $NC$ in this analysis. Define $S:=q_1+q_2$.
The solution is given by: 
\begin{equation}
\partial \Pi_1(q_1,q_2)/\partial q_1 = 0 \,, 
\end{equation}
that is, 
\begin{equation}
\label{eq:2N_1}
S^{-(1+(1/\eps))}(S-q_1/\eps) = \beta_1 Q^{-1/\eps} \,,
\end{equation}
and
\begin{equation}
\partial \Pi_2(q_1,q_2)/\partial q_2 = 0 \,, 
\end{equation}
that is, 
\begin{equation}
\label{eq:2N_2}
S^{-(1+(1/\eps))}(S-q_2/\eps) = \beta_2 Q^{-1/\eps} \,. 
\end{equation}
For any fixed $q_2$, we can solve Equation~\ref{eq:2N_1} according to:
\begin{eqnarray*}
(q_1+q_2)^{-(1+\frac{1}{\eps})}(q_1(1-\frac{1}{\eps})+q_2)&=&\beta_1Q^{-\frac{1}{\eps}}\\
\Leftrightarrow (q_1+q_2)^{1+\eps}(q_1(1-\frac{1}{\eps})+q_2)^{-\eps}&=&\beta_1^{-\eps}Q\\
\Leftrightarrow
q_2^{1+\eps}(\frac{q_1}{q_2}+1)^{1+\eps}q_2^{-\eps}(\frac{q_1}{q_2}(1-\frac{1}{\eps})+1)^{-\eps}&=&\beta_1^{-\eps}Q\\
\Leftrightarrow 
\left(
\frac{(\frac{q_1}{q_2}+1)^{1+\frac{1}{\eps}}}{\frac{q_1}{q_2}(1-\frac{1}{\eps})+1}
\right)&=&\left(\frac{Q}{q_2\beta_1^{\eps}}\right)^{1/\eps}.
\end{eqnarray*}
In other words, $q_1/q_2$ is the solution to the equation,
$$
(z+1)^{1+\frac{1}{\eps}}=Az+B,
$$
where $A=(1-\frac{1}{\eps})(Q/(q_2\beta_1^{\eps}))^{1/\eps}$ and $B=(Q/(q_2\beta_1^{\eps}))^{1/\eps}$.
We use $\hat{q}_1(q_2)$ to denote this solution for any fixed value of
$q_2$. 

For any fixed $q_1$, we can solve Equation~\ref{eq:2N_2} according to:
\begin{eqnarray*}
(q_1+q_2)^{-(1+\frac{1}{\eps})}(q_1+q_2(1-\frac{1}{\eps}))&=&\beta_2Q^{-\frac{1}{\eps}}\\
\Leftrightarrow (q_1+q_2)^{1+\eps}(q_1+q_2(1-\frac{1}{\eps}))^{-\eps}&=&\beta_2^{-\eps}Q\\
\Leftrightarrow
q_1^{1+\eps}(1+\frac{q_2}{q_1})^{1+\eps}q_1^{-\eps}(1+\frac{q_2}{q_1}(1-\frac{1}{\eps}))^{-\eps}&=&\beta_2^{-\eps}Q\\
\Leftrightarrow 
\left(
\frac{(1+\frac{q_2}{q_1})^{1+\frac{1}{\eps}}}{1+\frac{q_2}{q_1}(1-\frac{1}{\eps})}
\right)&=&\left(\frac{Q}{q_1\beta_2^{\eps}}\right)^{1/\eps}.
\end{eqnarray*}
In other words, $q_2/q_1$ is the unique positive solution to the equation,
$$
(1+z)^{1+\frac{1}{\eps}}=Az+B,
$$
where $A=(1-\frac{1}{\eps})(Q/(q_1\beta_2^{\eps}))^{1/\eps}$ and $B=(Q/(q_1\beta_2^{\eps}))^{1/\eps}$.
We use $\hat{q}_2(q_1)$ to denote this solution for any fixed value of
$q_1$. 

Suppose we now want to solve both Equations~\ref{eq:2N_1} and
\ref{eq:2N_2} simultaneously. By dividing the two equations we have 
\begin{eqnarray*}
\beta_2(S-\frac{q_1}{\eps}) &=& \beta_1(S-\frac{q_2}{\eps})\\ 
\Rightarrow \frac{\beta_2}{\beta_1}(q_1(1-\frac{1}{\eps})+q_2) &=&
q_1+q_2(1-\frac{1}{\eps})\\
\Rightarrow (1-\frac{\beta_2}{\beta_1}(1-\frac{1}{\eps}))q_1&=&(\frac{\beta_2}{\beta_1}-(1-\frac{1}{\eps}))q_2.
\end{eqnarray*}
For $\beta_2/\beta_1 \neq (1-1/\eps)$, we have
\begin{equation}
q_2 = \frac{1- (\beta_2/\beta_1) (1-1/\eps)}{\beta_2/\beta_1-(1-1/\eps)}q_1
\,,
\end{equation}
and so $t:=q_2/q_1=\frac{1-\beta_2/\beta_1
  (1-1/\eps)}{\beta_2/\beta_1-(1-1/\eps)}$, 
which can be calculated by knowing the parameters $\eps$,
$\beta_1,\beta_2$. 
If we let $(q_1^*,q_2^*)$ be the solution to the simultaneous equations
then it follows,
\begin{equation}
\label{eq:q1nc}
q_1^*=
Q\left(\frac{1+t(1-\frac{1}{\eps})}{\beta_2(1+t)^{1+\frac{1}{\eps}}}\right)^{\eps}
\,,
\end{equation}
and
\begin{equation}
\label{eq:q2nc}
q_2^*=Q\left(
\frac{\frac{1}{t}(1-\frac{1}{\eps})+1}{\beta_1(\frac{1}{t}+1)^{1+\frac{1}{\eps}}}
\right)^{\eps}
\,.
\end{equation}
(The ratio $q_2^*/q_1^*$ indeed equals $t$). The corresponding profits
are given by,
\begin{equation}
\Pi_1^{NC}:=\Pi_1(q_1^*,q_2^*)= \left(\frac{Q}{q_1^*+q_2^*}\right)^{1/\eps}
q_1^* - \alpha_1 -\beta_1 q_1^*\,,
\end{equation}
and
\begin{equation}
\Pi_2^{NC}:=\Pi_2(q_1^*,q_2^*)= \left(\frac{Q}{q_1^*+q_2^*}\right)^{1/\eps}
q_2^* - \alpha_2 -\beta_2 q_2^*\,.
\end{equation}
\end{proof}

We now consider the real, i.e.\ non-relaxed problem and derive the
conditions under which the above Nash equilibrium is a valid
solution. This is the case if $q_1^*$, $q_2^*$, $\Pi_1(q_1^*,q_2^*)$
and $\Pi_2(q_1^*,q_2^*)$ are all non-negative. From the above analysis
we immediately have,
\begin{lemma}
\label{l:nash-viable}
The pair $(q_1^*,q_2^*)$ is a viable solution to the original problem
if and only if,
\begin{eqnarray*}
(1-1/\eps) &\le& \min\{\beta_1/\beta_2,\beta_2/\beta_1\},\\
\alpha_1&\le&\left(\frac{Q}{q_1^*+q_2^*}\right)^{1/\eps}q_1^*-\beta_1 q_1^*,\\
\alpha_2&\le&\left(\frac{Q}{q_1^*+q_2^*}\right)^{1/\eps}q_2^*-\beta_2 q_2^*.
\end{eqnarray*}
\end{lemma}
Suppose that the conditions of Lemma~\ref{l:nash-viable} do not hold. Note that due to
the unimodal nature of the profit curve, if $q_i^*$ is negative then
SP $i$ would be better off setting $q_i=0$. Also, if $\Pi_i^{NC}$ is
negative then SP $i$ would be better off setting $q_i=0$. In each of
these cases we say that SP $i$ is {\em driven out of the
  market}. 

However, as we first discussed in
Section~\ref{s:aggressive-submissive}, even if the Nash equilibrium is
a viable solution an SP may still be incentivized to try and drive
the other SP out of the market. 
other SP out of the market and so we now investigate that situation in
detail. In particular let, 
$$
q'_1=\arg\max_{q_1:\Pi_2(q_1,\hat{q}_2(q_1))\le 0}\{\Pi_1(q_1,0)\}.
$$
In other words, let $q'_1$ be the value of $q_1$ that maximizes the
profit of SP1 assuming that it can drive SP2 out of the market
even if SP2 gives the best response. 
Similarly let, 
$$
q'_2=\arg\max_{q_2:\Pi_1(\hat{q}_1(q_2),q_2)\le 0}\{\Pi_2(0,q_2)\}.
$$
With these formulas in place, SP1 chooses from the following three options.
\begin{itemize}
\item Option 1 (Nash-Cournot). Playing value $q_1^*$ under the assumption that SP2
  plays value $q_2^*$. This option is viable if all of the quantities,
  $q_1^*$, $q_2^*$, $\Pi_1(q_1^*,q_2^*)$ and $\Pi_2(q_1^*,q_2^*)$ are
  non-negative, i.e.\ if the conditions of Lemma~\ref{l:nash-viable}
  hold.  
\item Option 2 (Aggression). Playing value $q'_1$ with the expectation that SP2
  plays value $0$, (i.e.\ it does not participate in the market). 
\item Option 3 (Submission). Playing value $0$, i.e.\ not participating in the
  market. 
\end{itemize}
SP2 is faced with an analogous set of options and so we obtain the
following profit table in Figure~\ref{f:payofftable-gen} that is a general
version of the first three columns and rows in Figure~\ref{f:payofftable}. 
\begin{figure*}[htb]
\centering
\begin{tabular}{|c|c|c|c|}\hline 
& Nash-Cournot & Aggression & Submission \\ \hline
Nash-Cournot&$(\Pi_1(q^*_1,q^*_2),\Pi_2(q^*_1,q^*_2))$
&$(\Pi_1(q^*_1,q'_2),\Pi_2(q^*_1,q'_2))$
&$(\Pi_1(q^*_1,0),0)$\\ \hline
Aggression&$(\Pi_1(q'_1,q^*_2),\Pi_2(q'_1,q^*_2)$
&$(\Pi_1(q'_1,q'_2),\Pi_2(q'_1,q'_2))$
&$(\Pi_1(q'_1,0),0)$ \\ \hline
Submission&$(0,\Pi_2(0,q^*_2))$
&$(0,\Pi_2(0,q'_2))$
&$(0,0)$\\ \hline
\end{tabular}
\caption{The profits $(\Pi_1,\Pi_2)$ due to the different strategy
  combinations. The rows represent the decisions for SP1 and the
  columns represent the decisions for SP2.}
\label{f:payofftable-gen}
\end{figure*}

\section{Network Sharing}
\label{s:sharing}

We now examine the situation under network sharing. In this case the
SPs cooperate and use the lowest cost parameters that are available,
i.e.\
$$
C^{coop}(q)=\alpha_{min}+\beta_{min}q,
$$
where $\alpha_{min}=\min\{\alpha_1,\alpha_2\}$ and
$\beta_{min}=\min\{\beta_1,\beta_2\}$.

We first assume that the combined entity is able to use monopoly
pricing. In this case the combined price, demand and profit is given
by,
\begin{eqnarray*}
p^{coop}&=&\ve \beta_{min}/(\ve-1)\\
q^{coop}&=&Q(\ve \beta_{min}/(\ve-1))^{-\ve}\\
\Pi^{coop}&=&p^{coop} q^{coop} - (\alpha_{min}+\beta_{min} q^{coop})\\
\end{eqnarray*}

It remains to determine how the profit is split between the SPs. A
natural way to do this is via the Shapley value which gives to SP $i$
its expected contribution to the coalition assuming that the SPs
create the coalition in a random order. There are two ways to
calculate this number depending on whether we incorporate {\em
  externalities}~\cite{MichalakRMSJ10,Myerson77} from outside the coalition. More precisely, when SP
$i$ is the first member of the coalition, we can either assume that it can
utilize monopoly pricing or we can assume that it still has to
compete against the other SP according to the Nash-Cournot game. The
former case might be more appropriate in a rural setting in which
an SP that enters the market late is unlikely to participate in the
market unless it can share. The latter case
might be more appropriate in an urban situation in which both SPs feel
compelled to enter the market regardless of whether or not they can
share. 
In
the first case we get,
\begin{eqnarray*}
\Pi_1^{coop}&=&\frac{1}{2}(\Pi_1^{mon}+\Pi^{coop}-\Pi_2^{mon})\\
\Pi_2^{coop}&=&\frac{1}{2}(\Pi_2^{mon}+\Pi^{coop}-\Pi_1^{mon}).
\end{eqnarray*}
In the latter case we get,
\begin{eqnarray*}
\Pi_1^{coop}&=&\frac{1}{2}(\Pi_1^{NC}+\Pi^{coop}-\Pi_2^{NC})\\
\Pi_2^{coop}&=&\frac{1}{2}(\Pi_2^{NC}+\Pi^{coop}-\Pi_1^{NC}).
\end{eqnarray*}
(In Figure~\ref{f:hierarchy} these cases make up the 
{\bf ``Cooperation$\rightarrow$No regulator$\rightarrow$Profit shared via
Shapley allocation''} branches.)

From a regulator's point of view, the downside of sharing under
monopoly pricing is that the price is significantly higher than the competitive
case. We now look at an alternative framework in which the price 
is restricted by a regulator to be the same as in the Nash-Cournot
game. In this setting the solution becomes,
\begin{eqnarray*}
p^{coop}&=&p^{NC}\\
q^{coop}&=&Q(p^{NC})^{-\ve}\\
\Pi^{coop}&=&q^{coop}p^{coop}-(\alpha_{min}+\beta_{min}q^{coop}).
\end{eqnarray*}
Now when we consider the Shapley value without externalities, it only
makes sense to assume that the price is constrained to be the
regulated price, regardless of the size of the coalition. Hence in all
cases the price and demand are the same and so the only difference
between the coalitions is the cost. 
\begin{eqnarray*}
\Pi_1^{coop}&=&\frac{1}{2}(q^{coop}p^{coop}-(\alpha_1+\alpha_{min}-\alpha_2)\\&&-(\beta_1+\beta_{min}-\beta_2)q^{coop})\\
\Pi_2^{coop}&=&\frac{1}{2}(q^{coop}p^{coop}-(\alpha_2+\alpha_{min}-\alpha_1)\\&&-(\beta_2+\beta_{min}-\beta_1)q^{coop})
\end{eqnarray*}
For the Shapley value with externalities we have,
\begin{eqnarray*}
\Pi_1^{coop}&=&\frac{1}{2}((q_1^{NC}+q^{coop}-q_2^{NC})p^{coop}-(\alpha_1+\alpha_{min}-\alpha_2)\\&&-(\beta_1q_1^{NC}+\beta_{min}q^{coop}-\beta_2q_2^{NC}))\\
\Pi_2^{coop}&=&\frac{1}{2}((q_2^{NC}+q^{coop}-q_1^{NC})p^{coop}-(\alpha_2+\alpha_{min}-\alpha_1)\\&&-(\beta_2q_2^{NC}+\beta_{min}q^{coop}-\beta_1q_1^{NC}))
\end{eqnarray*}

\section{Price competition and breaking the ``Bertrand curse''}
\label{s:bertrand-single}
The results of Appendix~\ref{s:cournot} were for the Cournot notion of
competition by offered capacity. Each SP decides to serve demand $q_i$ and
then the price is determined by the total capacity $q_1+q_2$. The main
alternative notion of competition is a {\em Bertrand} competition in
which the providers offer a price to the market. In this section we
examine how a price-based competition operates under various notions of price-sensitivity for the
end users. We start with the outcome of our running example for the
case of a Bertrand competition. 

\subsection{Bertrand analysis for the running example.}

Recall our example from Section~\ref{s:narrative} in which the price elasticity function has parameters $\ve=1.25$ and
$Q=1000$.  The per-unit capacity costs are $\beta_1=\$2.5M$ for SP1
and $\beta_2=\$2M$ for SP2 and the fixed
capacity costs 
are $\alpha_1 = \$50M$ and $\alpha_2 = \$100M$.

We now consider the dynamics of this example under the simplest type of Bertrand game. If the
SPs offer different prices then all the demand goes to the one with
the lowest price. If the two SPs offer the same price then the demand
is split between them. The value of the lowest price determines the
amount of demand in the market. 
\begin{figure*}[htb]
\begin{center}
\includegraphics[width=2.6in]{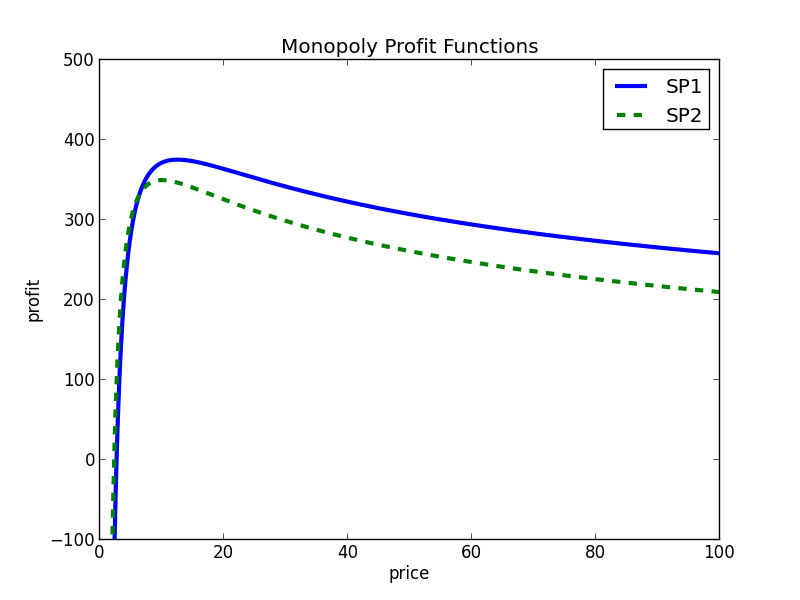}
~~
\includegraphics[width=2.6in]{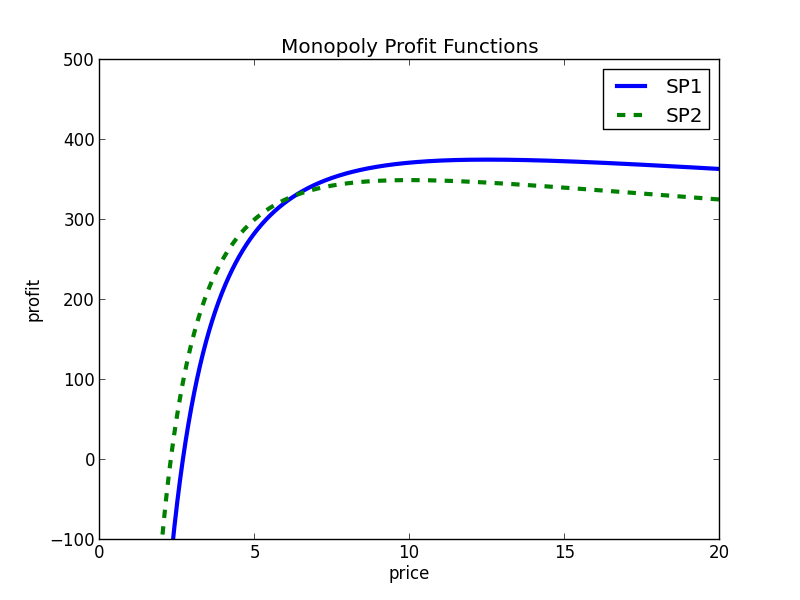}
\caption{(Left) The profit for the two SPs as a function of the price $p$. Here $\alpha_1=\$50M$,
  $\beta_1=\$2.5M$, $\alpha_2=\$100M$ and $\beta_2=\$2M$. (Right) A
  blown up version of the figure.}
\label{f:monopolyqp}
\end{center}
\end{figure*}
In Figure~\ref{f:monopolyqp} (left) we show the monopoly profit function for these
parameters as a function of price, $p$. We denote these functions by $\Pi_1^{mon}(p)$
and $\Pi_2^{mon}(p)$.  
Figure~\ref{f:monopolyqp} (right) shows a blown up version of
the figure.

We see from the figure that $\Pi_1^{mon}(p)\le 0$ for $p\le 2.68$ and
$\Pi_2^{mon}(p)\le 0$ for $p\le 2.28$. Hence if $p_1\ge 2.68$ then SP2 can always claim all the demand and
gain a positive profit by setting $p_2=\frac{1}{2}(p_1+2.28)$. As a
result, a natural strategy for SP2 is to set $p_2$ slightly under
2.68 (e.g.\ at $2.67$). This effectively drives SP1 out of the market since there is no
setting for $p_1$ that would allow it to claim nonzero demand and make
a positive profit. The resulting solution is:
\begin{eqnarray*}
p_1=\$2.68M,&~&q_1=0PB,~\Pi_1=\$0M\\
p_2=\$2.67M,&~&q_2=293PB,~\Pi_2=\$96.3M.
\end{eqnarray*}

We now comment on the difference between the Bertrand and Cournot
results and how they compare to the sharing scenario. 
\begin{itemize}
\item In the Cournot game, each SP has the ability to drive the
  other out of the market. In the Bertrand game, only SP2 can
  do that. The reason for this asymmetry in the Bertrand game is that
  SP2 has a lower variable cost, i.e. $\beta_2<\beta_1$. (For example,
  this might be due to SP2 having a lower cost for deploying capacity.) Hence
  SP2 can make a profit at a lower price than SP1. 
\item If competition is modeled according to a Bertrand game, SP2 is better off sharing with SP1 than
  driving SP1 out of the market. In contrast, if competition is
  modeled according to a Cournot game, both SPs can do better 
  than sharing if they are able to drive the other SP out of the
  market. 
\item With the Bertrand game, if SP2 drives SP1 out of the market,
  there is no action of SP1 that would cause SP2 to have a negative
  profit. In contrast, in the Cournot game if one SP is aggressive and
  tries to drive the other out of the market, the aggressive SP is in
  danger of making a negative profit if the other SP refuses to be
  submissive and also decides to act aggressively.
\end{itemize}

We now provide a general analysis of the Bertrand game.
In Section~\ref{s:bertrand1} we examine the basic Bertrand model 
and then in
Section~\ref{s:bertrand-varian} we show how the results change when
not all users are price sensitive. This type of model has been
considered as one way of breaking the ``Bertrand curse'' which occurs
when both providers have the same costs and so neither can achieve a profit. 
In Section~\ref{s:sharing-price-competition} we 
consider another extension of the basic model that is motivated by
potential actions of a regulator.  
In this extension (that is a combination
of the sharing model and the Bertrand game) the SPs are
allowed to share costs but the regulator enforces that they must still compete on price. 

\subsection{Bertrand model 1: all end users are price sensitive}
\label{s:bertrand1}
In this simplest case we assume that all demand goes to the SP that
offers the lowest price and if both SPs offer the same price then
the demand is split between them. Suppose that the $\alpha_i,\beta_i$
parameters are fixed.
\begin{itemize}
\item $\alpha_i\ge
  Q((\ve\beta_i/(\ve-1))^{1-\ve}-\beta_i^{1-\ve}(\ve/(\ve-1))^{-\ve})$
  for $i=1,2$. 
In this case a Nash equilibrium is for both SPs to stay out of the
market and keep a profit of zero. This is because neither can attain a
positive profit, even if they can act as a monopoly. 
\item If the above condition does not hold for SP $i$, let
  $\underline{p}_i=\min\{p:Qp^{1-\ve}-(\alpha_i+\beta_iQp^{-\ve})\ge
  0\}$ and let 
  $\bar{p}_i=\max\{p:Qp^{1-\ve}-(\alpha_i+\beta_iQp^{-\ve})\ge
  0\}$. Suppose without loss of generality that $\underline{p}_1\ge
  \underline{p}_2$. 
  If $\bar{p}_2\le \underline{p}_1$ then a Nash equilibrium is for
  SP2 to set price $p_2=p_2^{mon}$ and for SP1 to stay out of the
  market. 
\item If $\bar{p}_2> \underline{p}_1$ and  $\underline{p}_2<
  \underline{p}_1$ (i.e.\ with strict inequality) then a natural 
  solution is for SP2 to set price
  $p_2=\min{p_2^{mon},\underline{p}_1}$ and for SP1 to stay out of the
  market. Note that this is not a Nash equilibrium in the strict sense
  since if SP1 does not participate then the optimal action of SP2
  is to set its price to its monopoly price. However, if SP2 did that
  then SP1 could potentially get back into the market. Hence a more
  natural course of action for SP2 is to set its price to the best
  price that keeps SP1 out of the market. 
\item If $\underline{p}_2=\underline{p}_1$ then it is not hard to see
  that the only Nash equilibrium is for both SPs to set price
  $p_i=\underline{p}_i$ which gives them zero profit. This is an
  example of the 
  so-called {\em Bertrand curse} in which either one provider is driven
  completely out of the market or else both providers make zero
  profit. 
\end{itemize}


\subsection{Bertrand model 2: not all end users are price sensitive}
\label{s:bertrand-varian}
The Bertrand curse is generally viewed as a an undesirable state of
affairs and so there has been much research on methods to avoid it. We
now consider how our results change in a framework of Bagwell and Lee~\cite{BagwellL14}
that was inspired by earlier work Varian~\cite{Varian80}. 
In this model we assume that a fraction $I$ of the end users (the ``informed''
users) are price sensitive and a fraction $U=1-I$ (the ``uninformed''
users) are price insensitive. However, each end user still generates
demand based on price according to the price elasticity function. 
Let,
\begin{eqnarray*}
\Pi_i^+(p)&=&(I+\frac{U}{2})pq(p)-(\alpha_i+\beta_iq(p))\\
\Pi_i^-(p)&=&\frac{U}{2}pq(p)-(\alpha_i+\beta_iq(p))
\end{eqnarray*}
In other words, $\Pi_i^+(p)$ is the profit function when SP $i$ has
the low price and $\Pi_i^-(p)$ is the profit function when SP $i$ has
the high price. 
Let $[\underline{p}^-_i,\bar{p}^-_i]$ be the price range on which
$\Pi^-_i$ is non-negative and let $[\underline{p}^+_i,\bar{p}^+_i]$ be the price range on which
$\Pi^+_i$ is non-negative. We assume without loss of generality that
$\underline{p}^+_2\le \underline{p}^+_1$. 

For this case we claim that the following is a stable solution. SP1
sets its price $p_1$ so as to maximize $\Pi^-_1$. Now let
$\hat{p}_1\le p_1$
be such that $\Pi^+_1(\hat{p}_1)=\Pi^-_1(p_1)$. SP2 sets its price
$p_2=\min\{\hat{p}_1,\arg\max\Pi^+_2(p)\}$.

We show that this solution is a stable solution in the following
sense. First of all, given the price offered by SP2, SP1 cannot
improve its profit with any other price. Hence it satisfies the
property of a Nash equilbrium from the perspective of SP1. It does
not satsify the property of a Nash equilibrim from the perspective of
SP2 since SP2 could potentially increase its profit if SP1
keeps its price fixed. However, if SP2 did that then SP1 could
choose a new price in which SP1 does better and SP2 does
worse. Another way to look at this is that these prices form a subgame
perfect Nash equilibrium for the 
Stackelberg game in which SP2 sets its price first and then SP1
follows. 


\subsection{Sharing on cost with price competition}
\label{s:sharing-price-competition}
When considering the benefits of sharing in the context of a Bertrand
competition, one model of sharing would be exactly as was considered
before in the context of a Nash Cournot competition. The service
providers provide capacity based on the minimum of their costs and
then calculate a monopoly price with respect to those costs. The above
notion of a Bertrand competition with insensitive users gives rise to
another notion of sharing that might be more appealing to a regulator.
In particular, the SPs are allowed to cooperate on cost when building
capacity. However, when offering service to end users they must still
compete on price. This gives rise to a Bertrand competition in which
each SP has parameters $\alpha_{\min}=\min\{\alpha_1,\alpha_2\}$ and
$\beta_{\min}=\min\{\beta_1,\beta_2\}$. In the case of a Bertrand competition in
which all users are price sensitive, both SPs would offer a price
$p=\beta_{\min}$ and so neither would generate a profit. However, for the
case in which some users are price insensitive, there is a stable
situation as described above in which one SP offers a low price in
order to get all the price sensitive users while the other one offers
a high price in order to get all the price insensitive users. 
%

\section{Multiple Geographic Regions}
\label{s:narrative-multiple}

One of the main reasons that
regulators allow network sharing even though it leads to loss of
competition is that it allows service providers to more quickly offer
service over a large geographic region. In order to investigate this
phenomenon, we now show how our analysis extends when it is not the
case that both SPs can offer service by themselves over the entire
market. For this analysis we focus on the specific SP cost parameters
given in Section~\ref{s:narrative}.

We consider two scenarios, both of which have two regions W and E. In
the first scenario SP1 can provide
service in region W (with its $\alpha$ value halved to represent fixed
costs in one region only), and SP2 can provide service in region
E (with its $\alpha$ value also halved). There are three types of user, W, E and WE. WE users need a service
provider that can provide service in both regions. Hence if the SPs do
not cooperate then these users cannot be served. Let $N_X$ be the
fraction of users of type $X$. For our numerical example we assume
that $N_W=N_E=N_{WE}=\frac{1}{3}$. 

In the case without sharing there is no competition and we have,
\begin{eqnarray*}
p_1^{mon}=\$12.5M,&~&q_1^{mon}=14.18PB,~\Pi_1^{mon}=\$116.82M\\
p_2^{mon}=\$10M,&~&q_2^{mon}=18.74PB,~\Pi_2^{mon}=\$99.96M
\end{eqnarray*}

In the case of cooperation we assume that each SP builds the network
in its ``own'' region but the two SPs cooperate on price over the
entire market. Hence the combined entity has a monopoly with parameters
$(\alpha_1,\beta_1)$ in the W region and a monopoly with parameters
$(\alpha_2,\beta_2)$ in the E region. Thus, 
$$
p^{coop}=\$11.25M,~q^{coop}=48.54PB,~\Pi^{coop}=\$362M
$$
Hence from the Shapley value we have,
\begin{eqnarray*}
\Pi_1^{coop}&=&\frac{1}{2}(116.82+362-99.96)=\$189.43M\\
\Pi_2^{coop}&=&\frac{1}{2}(99.96+362-116.82)=\$172.57M
\end{eqnarray*}

In the second scenario we assume initally that SP1 can only serve users in the W
region but SP2 can serve users in both the W and E regions. Hence SP2
has a monopoly on both the E users and the WE users. For these two
sets of users we have,
$$
p_2=\$10M,~q_2=37.49PB
$$
For the W users we have a competition between the SPs. If it is a
Nash-Cournot competition then we have:
\begin{eqnarray*}
p_1^{NC}=\$3.75M,&~&q_1^{NC}=26.62PB\\
p_2^{NC}=\$3.75M,&~&q_2^{NC}=37.26PB
\end{eqnarray*}
The total profit in this case is:
\begin{eqnarray*}
\Pi_1^{NC}=\$8.28M,&~&\Pi_2^{NC}=\$265M
\end{eqnarray*}
If it is a Bertrand competition then SP2 is always incentivized to
compete since it acts as a monopoly for the E and WE users. Hence it
sets its price to the minimum that drives SP1 out of the market for
the W users, i.e.\ we have,
\begin{eqnarray*}
p_1^B=\$3.75M,&~&q_1^B=0PB,~\Pi_1^B=\$0M\\
p_2^B\in\{\$2.77M,\$10M\},&~&q_2^B=130.77PB,~\Pi_2^B=\$272M
\end{eqnarray*}
where $p_2^B=\$2.77M$ for the W users and $p_2^B=\$10M$ for the WE and
E users. 
For the cooperative solution the SPs have to use the SP2 parameters
in the E region but they can use the optimum of the SP1 and SP2
parameters in the W region. Hence in this case we have,
$$
p^{coop}=\$11.25M,~q^{coop}=48.54PB,~\Pi^{coop}=\$362M.
$$
In order to calculate the Shapley value to split the profit we need to
know the individual monopoly values in this case. 
\begin{eqnarray*}
p_1^{mon}=\$12.5M,&~&q_1^{mon}=14.18PB,~\Pi_1^{mon}=\$117M\\
p_2^{mon}=\$10.0M,&~&q_2^{mon}=56.2PB,~\Pi_2^{mon}=\$350M.
\end{eqnarray*}
Hence,
\begin{eqnarray*}
\Pi_1^{coop}&=&\frac{1}{2}(117+362-350)=\$64.5M\\
\Pi_2^{coop}&=&\frac{1}{2}(350+362-117)=\$297.5M
\end{eqnarray*}

\section{Proof of Lemma~\ref{l:monopoly}}
\label{s:monopoly}
\begin{proof}
In the following we drop the $mon$ superscript. 
We determine the
solution by setting $d\Pi/dq=0$. Recall our assumption that $\ve>1$.\footnote{We
  assume $\ve>1$ since otherwise the SPs would be incentivized to set
  an arbitrarily high price.}  
\begin{eqnarray*}
\Pi &=& \frac{q^{1-(1/\ve)}}{Q^{-1/\ve}}-(\alpha+\beta q)\\
\frac{d\Pi}{dq} &=& \frac{\ve-1}{\ve}\frac{q^{-1/\ve}}{Q^{-1/\ve}}-\beta.
\end{eqnarray*}
Hence $\frac{d\Pi}{dq}=0$ if and only if,
\begin{eqnarray*}
\frac{q^{-1/\ve}}{Q^{-1/\ve}}&=&\beta\frac{\ve}{\ve-1}\\
\Leftrightarrow q &=& Q\left(\frac{\ve\beta}{\ve-1}\right)^{-\ve}\\
\Leftrightarrow p &=&
\left(\frac{Q}{Q}\left(\frac{\beta\ve}{\ve-1}\right)^{-\ve}\right)^{-1/\ve}\\
\Leftrightarrow p &=& \frac{\beta\ve}{\ve-1}
\end{eqnarray*}
\end{proof}

\end{document}